\documentclass[11pt]{article}

\usepackage[margin=1in]{geometry}
\usepackage{amsthm,amsmath,amsfonts,amssymb}
\usepackage{url}
\usepackage{xspace}
\usepackage{fancyhdr}

\usepackage{hyperref}
\hypersetup{
    bookmarksnumbered=true, 
    unicode=false, 
    pdfstartview={FitH}, 
    pdftitle={The quantum query complexity of read-many formulas}, 
    pdfauthor={Andrew M. Childs, Shelby Kimmel, and Robin Kothari}, 
    pdfsubject={}, 
    pdfcreator={}, 
    pdfproducer={}, 
    pdfkeywords={}, 
    pdfnewwindow=true, 
    colorlinks=true, 
    linkcolor=blue, 
    citecolor=blue, 
    filecolor=blue, 
    urlcolor=blue 
}

\newtheorem{theorem}{Theorem}
\newtheorem{lemma}{Lemma}
\newtheorem{proposition}{Proposition}
\newtheorem{corollary}{Corollary}

\newcommand{\eq}[1]{\hyperref[eq:#1]{(\ref*{eq:#1})}}

\renewcommand{\sec}[1]{\hyperref[sec:#1]{Section~\ref*{sec:#1}}}
\newcommand{\thm}[1]{\hyperref[thm:#1]{Theorem~\ref*{thm:#1}}}
\newcommand{\lem}[1]{\hyperref[lem:#1]{Lemma~\ref*{lem:#1}}}
\newcommand{\prop}[1]{\hyperref[prop:#1]{Proposition~\ref*{prop:#1}}}
\newcommand{\cor}[1]{\hyperref[cor:#1]{Corollary~\ref*{cor:#1}}}
\newcommand{\fig}[1]{\hyperref[fig:#1]{Figure~\ref*{fig:#1}}}

\renewcommand{\th}[1]{${#1}^{\textrm{th}}$}

\renewcommand{\(}{\left(}
\renewcommand{\)}{\right)}

\newcommand{\ACzero}{$\mathrm{AC}^0$\xspace}
\newcommand{\ACzerob}{$\mathrm{AC}^0_{\mathrm{b}}$\xspace}

\begin{document}

\title{The quantum query complexity of read-many formulas}

\author{
\normalsize Andrew M.\ Childs\thanks{amchilds@uwaterloo.ca} \\[.5ex]
\small Department of Combinatorics \& Optimization \\
\small and Institute for Quantum Computing \\
\small University of Waterloo
\and
\normalsize Shelby Kimmel\thanks{skimmel@mit.edu} \\[.5ex]
\small Center for Theoretical Physics \\
\small Massachusetts Institute of Technology
\and
\normalsize Robin Kothari\thanks{rkothari@cs.uwaterloo.ca} \\[.5ex]
\small David R.\ Cheriton School of Computer Science \\
\small and Institute for Quantum Computing \\
\small University of Waterloo
}

\date{}
\maketitle

\thispagestyle{fancy}
\rhead{MIT-CTP 4330}
\renewcommand{\headrulewidth}{0pt}
\renewcommand{\footrulewidth}{0pt}

\begin{abstract}
The quantum query complexity of evaluating any read-once formula with $n$ black-box input bits is $\Theta(\sqrt n)$.  However, the corresponding problem for read-many formulas (i.e., formulas in which the inputs have fanout) is not well understood.  Although the optimal read-once formula evaluation algorithm can be applied to any formula, it can be suboptimal if the inputs have large fanout.  We give an algorithm for evaluating any formula with $n$ inputs, size $S$, and $G$ gates using $O(\min\{n, \sqrt{S}, n^{1/2} G^{1/4}\})$ quantum queries.  Furthermore, we show that this algorithm is optimal, since for any $n,S,G$ there exists a formula with $n$ inputs, size at most $S$, and at most $G$ gates that requires $\Omega(\min\{n, \sqrt{S}, n^{1/2} G^{1/4}\})$ queries.  We also show that the algorithm remains nearly optimal for circuits of any particular depth $k \ge 3$, and we give a linear-size circuit of depth $2$ that requires $\tilde\Omega(n^{5/9})$ queries.  Applications of these results include a $\tilde\Omega(n^{19/18})$ lower bound for Boolean matrix product verification, a nearly tight characterization of the quantum query complexity of evaluating constant-depth circuits with bounded fanout, new formula gate count lower bounds for several functions including \textsc{parity}, and a construction of an \ACzero circuit of linear size that can only be evaluated by a formula with $\Omega(n^{2-\epsilon})$ gates.
\end{abstract}

\section{Introduction}
\label{sec:intro}

A major problem in query complexity is the task of evaluating a Boolean formula on a black-box input.  In this paper, we restrict our attention to the standard gate set $\{\textsc{and}, \textsc{or}, \textsc{not}\}$.  A formula is a rooted tree of \textsc{not} gates and unbounded-fanin \textsc{and} and \textsc{or} gates, where each leaf represents an input bit and each internal vertex represents a logic gate acting on its children.  The depth of a formula is the length of a longest path from the root to a leaf.

The quantum query complexity of evaluating read-once formulas is now well understood.  (In this paper, the ``query complexity'' of $f$ always refers to the bounded-error quantum query complexity, denoted $Q(f)$.)  A formula is read-once if each input bit appears at most once in it. Grover's algorithm shows that a single \textsc{or} gate on $n$ inputs can be evaluated in $O(\sqrt n)$ queries \cite{Gro97}, which is optimal \cite{BBBV97}.  It readily follows that balanced constant-depth read-once formulas can be evaluated in $\tilde O(\sqrt n)$ queries \cite{BCW98} (where we use a tilde to denote asymptotic bounds that neglect logarithmic factors), and in fact $O(\sqrt n)$ queries are sufficient \cite{HMW03}.  A breakthrough result of Farhi, Goldstone, and Gutmann showed how to evaluate a balanced binary \textsc{and}-\textsc{or} formula with $n$ inputs in time $O(\sqrt n)$ \cite{FGG08} (in the Hamiltonian oracle model, which easily implies an upper bound of $n^{1/2 + o(1)}$ queries \cite{CCJY07}).  Subsequently, it was shown that any read-once formula whatsoever can be evaluated in $n^{1/2 + o(1)}$ queries \cite{ACRSZ10}, and indeed $O(\sqrt n)$ queries suffice \cite{Rei11}.  This result is optimal, since any read-once formula requires $\Omega(\sqrt n)$ queries to evaluate \cite{BS04}.

Now that the quantum query complexity of evaluating read-once formulas is tightly characterized, it is natural to consider the query complexity of evaluating more general formulas, which we call ``read-many'' formulas to differentiate them from the read-once case.   Note that we cannot expect to speed up the evaluation of arbitrary read-many formulas on $n$ inputs, since such formulas are capable of representing arbitrary functions, and some functions, such as parity, require as many quantum queries as classical queries up to a constant factor \cite{BBCMW01,FGGS98}.  Thus, to make the task of read-many formula evaluation nontrivial, we must take into account other properties of a formula besides the number of inputs.

Two natural size measures for formulas are formula size and gate count. The size of a formula, which we denote by $S$, is defined as the total number of inputs counted with multiplicity (i.e., if an input bit is used $k$ times it is counted as $k$ inputs). The gate count, which we denote by $G$, is the total number of \textsc{and} and \textsc{or} gates in the formula. (By convention, \textsc{not} gates are not counted.) Note that $G<S$ since for a given value of $S$, $G$ is largest when the formula is a binary tree with $S$ leaves, but the number of internal vertices in a binary tree is one less than the number of leaves. 

A formula of size $S$ can be viewed as a read-once formula on $S$ inputs by neglecting the fact that some inputs are identical, thereby giving a query complexity upper bound from the known formula evaluation algorithms. Thus, an equivalent way of stating the results on the evaluation of read-once formulas that also applies to read-many formulas is the following:

\begin{theorem}[Formula evaluation algorithm~{\cite[Corollary 1.1]{Rei11}}]\label{thm:formeval}
The bounded-error quantum query complexity of evaluating a formula of size $S$ is $O(\sqrt{S})$.
\end{theorem}

However, this upper bound does not exploit the fact that some inputs may be repeated, so it can be suboptimal for formulas with many repeated inputs. We study such formulas and tightly characterize their query complexity in terms of the number of inputs $n$ (counted \emph{without} multiplicity), the formula size $S$, and the gate count $G$.  By preprocessing a given formula using a combination of classical techniques and quantum search before applying the algorithm for read-once formula evaluation, we show that any read-many formula can be evaluated in $O(\min\{n, \sqrt{S}, n^{1/2} G^{1/4}\})$ queries (\thm{alg}).  Furthermore, we show that for any $n,S,G$, there exists a read-many formula with $n$ inputs, size at most $S$, and at most $G$ gates such that any quantum algorithm needs $\Omega(\min\{n, \sqrt{S}, n^{1/2} G^{1/4}\})$ queries to evaluate it (\thm{tight}).  We construct these formulas by carefully composing formulas with known quantum lower bounds and then applying recent results on the behavior of quantum query complexity under composition \cite{Rei09,Rei11}.

To refine our results on read-many formulas, it is natural to consider the query complexity of formula evaluation as a function of the depth $k$ of the given formula in addition to its number of input bits $n$, formula size $S$, and gate count $G$.  Beame and Machmouchi showed that constant-depth formulas (in particular, formulas of depth $k=3$ with size $\tilde O(n^2)$) can require $\tilde\Omega(n)$ queries \cite{BM10}.  However, to the best of our knowledge, nothing nontrivial was previously known about the case of depth-$2$ formulas (i.e., CNF or DNF expressions) with polynomially many gates, or formulas of depth $3$ and higher with gate count $o(n^2)$ (in particular, say, with a linear number of gates).

Building upon the results of \cite{BM10}, we show that the algorithm of \thm{alg} is nearly optimal even for formulas of any fixed depth $k \ge 3$.  While we do not have a tight characterization for depth-$2$ formulas, we improve upon the trivial lower bound of $\Omega(\sqrt n)$ for depth-$2$ formulas of linear gate count (a case arising in an application), giving an example of such a formula that requires $\tilde\Omega(n^{5/9}) = \Omega(n^{0.555})$ queries (\cor{depth2}).  It remains an open question to close the gap between this lower bound and the upper bound of $O(n^{0.75})$ provided by \thm{alg}, and in general, to better understand the quantum query complexity of depth-$2$ formulas.

Aside from being a natural extension of read-once formula evaluation, read-many formula evaluation has potential applications to open problems in quantum query complexity.  For example, the graph collision problem \cite{MSS07} can be expressed as a depth-$2$ read-many formula of quadratic gate count.
 
More concretely, we apply our results to better understand the quantum query complexity of Boolean matrix product verification.  This is the task of verifying whether $AB=C$, where $A,B,C$ are $n \times n$ Boolean matrices provided by a black box, and the matrix product is computed over the Boolean semiring, in which \textsc{or} plays the role of addition and \textsc{and} plays the role of multiplication.  Buhrman and \v{S}palek gave an upper bound of $O(n^{1.5})$ queries for this problem, and their techniques imply a lower bound of $\Omega(n)$ \cite{BS06}.  We improve the lower bound to $\tilde\Omega(n^{19/18}) = \Omega(n^{1.055})$ (\thm{bmpv}), showing in particular that linear query complexity is not achievable.

Our results can be viewed as a first step toward understanding the quantum query complexity of evaluating general circuits. A circuit is a directed acyclic graph in which source vertices represent black-box input bits, sink vertices represent outputs, and internal vertices represent \textsc{and}, \textsc{or}, and \textsc{not} gates. The size of a circuit is the total number of \textsc{and} and \textsc{or} gates, and the depth of a circuit is the length of a longest (directed) path from an input bit to an output bit.

The main difference between a formula and a circuit is that circuits allow fanout for gates, whereas formulas do not.  In a circuit, the value of an input bit can be fed into more than one gate, and the output of a gate can be fed into more than one subsequent gate. A read-once formula is a circuit in which neither gates nor inputs have fanout. A general formula allows fanout for the inputs, but not for the gates. Note that by convention, the size of a circuit is the number of gates, whereas the size of a formula is the total number of inputs counted with multiplicity, so one must take care to avoid confusion. For example, the circuit that has 1 \textsc{or} gate with $n$ inputs has circuit size $1$, formula size $n$, and formula gate count $1$.  To help clarify this distinction, we consistently use the symbol $S$ for formula size and the symbol $G$ for both formula gate count and circuit size.

As another preliminary result on circuit evaluation, we provide a nearly tight characterization of the quantum query complexity of evaluating constant-depth circuits with bounded fanout.  We show that any constant-depth bounded-fanout circuit of size $G$ can be evaluated in $O(\min\{n, n^{1/2} G^{1/4}\})$ queries, and that there exist such circuits requiring $\tilde\Omega(\min\{n, n^{1/2} G^{1/4}\})$ queries to evaluate.

Finally, our results on the quantum query complexity of read-many formula evaluation have two purely classical applications.  First, we give lower bounds on the number of gates (rather than simply the formula size) required for a formula to compute various functions.  For example, while it is known that parity requires a formula of size $\Omega(n^2)$ \cite{Khr71}, we show that in fact any formula computing parity must have $\Omega(n^2)$ gates, which is a stronger statement. We also give similar results for several other functions.  Second, for any $\epsilon>0$, we give an example of an explicit constant-depth circuit of linear size that requires $\Omega(n^{2-\epsilon})$ gates to compute with a formula of any depth (\thm{linsize}). 

The remainder of this paper is organized as follows.  In \sec{alg} we present the algorithm for evaluating read-many formulas, and in \sec{lb} we present a matching lower bound.  In \sec{const} we study the quantum query complexity of evaluating constant-depth formulas and constant-depth bounded-fanout circuits.  We present applications of these results in \sec{app}: the lower bound for Boolean matrix product verification appears in \sec{bmpv}, and results on classical circuit complexity appear in \sec{classical}.  Finally, we conclude in \sec{conc} with a discussion of the results and some open problems.

\section{Algorithm}
\label{sec:alg}

In this section, we describe an algorithm for evaluating any formula on $n$ inputs with ${G}$ gates that makes $O(n^{1/2} {G}^{1/4})$ queries to the inputs. When ${G} = \Omega(n^2)$, the upper bound is larger than $n$, so it is more efficient to simply read the entire input instead.  Similarly, if $n^{1/2}G^{1/4}$ is larger than $\sqrt{S}$, then it is preferable to use \thm{formeval}.  Thus, overall, we find an upper bound of $O(\min\{n, \sqrt{S}, n^{1/2} {G}^{1/4}\})$ queries.

In the next section we show that this algorithm is optimal in the sense that there exists a formula with $n$ inputs, size at most $S$, and at most ${G}$ gates that cannot be evaluated asymptotically faster. Note that this does not mean that all formulas with given values of $n,S,G$ have the same query complexity; formulas with appropriate structure can sometimes be evaluated much faster.  As an extreme example, the formula $x_1 \vee \bar x_1 \vee \cdots \vee x_n \vee \bar x_n$ is the constant function $1$, so it requires no queries to evaluate.

We first describe the algorithm at a high level. We wish to evaluate a formula $f$ with $n$ inputs and ${G}$ gates, using at most $O(n^{1/2} {G}^{1/4})$ queries. Consider directly applying the formula evaluation algorithm (\thm{formeval}). If the bottommost gates of $f$ have a large fanin, say $\Omega(n)$, then the formula size can be as large as $\Omega(n{G})$. If we apply the formula evaluation algorithm to this formula, we will get an upper bound of $O(\sqrt{n{G}})$, which is larger than claimed. So suppose that the formula size of $f$ is large.

Since the formula size is large, there must be some inputs that feed into a large number of gates, i.e., inputs with high degree. Among the inputs that feed into many \textsc{or} gates, if any input is 1, then this immediately fixes the output of a large number of \textsc{or} gates, reducing the formula size. A similar argument applies to inputs that are 0 and feed into \textsc{and} gates. Thus the first step in our algorithm is to find high-degree inputs and eliminate them, reducing the formula size considerably. Using at most $O(n^{1/2} {G}^{1/4})$ queries, we show how to eliminate enough high-degree inputs that the resulting formula has formula size $O(n\sqrt{G})$. Then we use \thm{formeval} on the resulting formula to achieve the claimed bound.

More precisely, our algorithm converts a formula $f$ on $n$ inputs and ${G}$ gates into another formula $f'$ of size $n\sqrt{G}$ on the same input. The new formula $f'$ has the same output as $f$ (on the given input), and this stage makes $O(n^{1/2}{G}^{1/4})$ queries. We call this the formula pruning algorithm.

Before explaining the algorithm, we need a subroutine to find a marked entry in a string of length $n$, with good expected performance when there are many marked entries.  One can find a marked entry with $O(\sqrt{n/t})$ expected queries where $t$ is the number of marked items, even when the value $t$ is not known~\cite{BBH+98}:

\begin{lemma} \label{lem:searchmany}
Given an oracle for a string $x \in \{0,1\}^n$, there exists a quantum algorithm that outputs with high probability the index of a marked item in $x$ if it exists, making $O(\sqrt{n/t})$ queries in expectation when there are $t$ marked items. If there are no marked items the algorithm runs indefinitely.
\end{lemma} 

We are now ready to state our formula pruning algorithm.

\begin{lemma}[Formula pruning algorithm]
\label{lem:pruning}
Given a formula $f$ with $n$ inputs and ${G}$ gates, and an oracle for the input $x$, there exists a quantum algorithm that makes $O(n^{1/2}{G}^{1/4})$ queries and returns a formula $f'$ on the same input $x$, such that $f'(x) = f(x)$ and $f'$ has formula size $O(n\sqrt{G})$.
\end{lemma}

\begin{proof}
First consider the set of all inputs that feed into \textsc{or} gates. From these inputs and \textsc{or} gates, we construct a bipartite graph with the inputs on one side and \textsc{or} gates on the other. We put an edge between an input bit and all of the \textsc{or} gates that take it as input. An input is called a {\em high-degree} input if it has degree greater than $\sqrt{G}$.

Now repeat the following process. First, use \lem{searchmany} to find any marked high-degree input. All \textsc{or} gates connected to this marked input have output 1, so we delete these gates and their input wires and replace the gates with the constant 1. This input is removed from the set of inputs for the next iteration. If all the \textsc{or} gates have been deleted or if there are no more high-degree inputs, the algorithm halts.

Say the process repeats $k-1$ times and in the \th{k} round there are no remaining marked high-degree inputs, although there are high-degree inputs. Then the process will remain stuck in the search subroutine (\lem{searchmany}) in the \th{k} round. Let the number of marked high-degree inputs at the \th{j} iteration of the process be $m_j$. Note that $m_j > m_{j+1}$ since at least 1 high-degree input is eliminated in each round. However it is possible that more than 1 high-degree input is eliminated in one round since the number of \textsc{or} gates reduces in each round, which affects the degrees of the other inputs. Moreover, in the last round there must be at least 1 marked item, so $m_{k-1} \geq 1$. Combining this with $m_j > m_{j+1}$, we get $m_{k-r} \geq r$.

During each iteration, at least $\sqrt{G}$ \textsc{or} gates are deleted. Since there are fewer than ${G}$ \textsc{or} gates in total, this process can repeat at most $\sqrt{G}$ times before we learn the values of all the \textsc{or} gates, which gives us $k \leq \sqrt{G}+1$.

In the \th{j} iteration, finding a marked input requires $O(\sqrt{n/m_j})$ queries in expectation by \lem{searchmany}. Thus the total number queries made in expectation until the \th{k} round, but not counting the \th{k} round itself, is
\begin{equation}
\sum_{j=1}^{k-1}O\(\sqrt{\frac{n}{m_j}}\) \leq  \sum_{j=1}^{k-1}O\(\sqrt{\frac{n}{j}}\)
\leq O(\sqrt{nk}) \leq O(n^{1/2}{G}^{1/4}).
\end{equation}

So in expectation, $O(n^{1/2}{G}^{1/4})$ queries suffice to reach the \th{k} round, i.e.,  to reach a stage where there are no high-degree marked inputs. To get an algorithm with worst-case query complexity $O(n^{1/2}{G}^{1/4})$, we simply halt this algorithm after it has made some constant times its expected number of queries. This gives a bounded-error algorithm with the same worst-case query complexity.

Next, we repeat the same process with \textsc{and} gates while searching for high-degree inputs that are zero. This also requires the same number of queries. At the end of both these steps, each input has at most $\sqrt{G}$ outgoing wires to \textsc{or} gates and at most $\sqrt{G}$ outgoing wires to \textsc{and} gates. This yields a formula $f'$ of size $O(n\sqrt{G})$ on the same inputs. 
\end{proof}

Combining this lemma with the formula evaluation algorithm gives the following.

\begin{theorem}\label{thm:alg}
The bounded-error quantum query complexity of evaluating a formula with $n$ inputs, size $S$, and ${G}$ gates is $O(\min\{n, \sqrt{S}, n^{1/2} {G}^{1/4}\})$.
\end{theorem}

\begin{proof}
We present three algorithms, with query complexities $O(n)$, $O(\sqrt S)$, and $O(n^{1/2} G^{1/4})$, which together imply the desired result.
Reading the entire input gives an $O(n)$ upper bound, and \thm{formeval} gives an upper bound of $O(\sqrt S)$.
Finally, we can use \lem{pruning} to convert the given formula to one of size $O(n\sqrt{G})$, at a cost of $O(n^{1/2} {G}^{1/4})$ queries.  \thm{formeval} shows that this formula can be evaluated using $O(n^{1/2} {G}^{1/4})$ queries.  \end{proof}

Let us return to the observation that our algorithm does better than the naive strategy of directly applying \thm{formeval} to the given formula with ${G}$ gates on $n$ inputs, since its formula size could be as large as $O(n{G})$, yielding a sub-optimal algorithm. Nevertheless, one might imagine that for every formula with ${G}$ gates on $n$ inputs, there exists another formula $f'$ that represents the same function and has formula size $n\sqrt{G}$. This would imply \thm{alg} directly using \thm{formeval}.

However, this is not the case: there exists a formula with ${G}$ gates on $n$ inputs such that any formula representing the same function has formula size $\Omega(n{G}/\log n)$.  This shows that in the worst case, the formula size of such a function might be close to $O(n{G})$. 

\begin{proposition}\label{prop:counting}
There exists a function $f$ that can be represented by a formula with ${G}$ gates on $n$ inputs, and any formula representing it must have formula size $S = \Omega(n{G}/\log n)$.
\end{proposition}

\begin{proof}[Proof sketch]
The proof is a counting argument. The total number of formulas of size $S$ is at most $\exp(O(S \log n))$, while there are at least $\exp(\Omega(n{G}))$ distinct functions with ${G}$-gate formulas. Thus $S = \Omega(n{G}/\log n)$.
\end{proof}

\section{Lower bounds}
\label{sec:lb}

We now turn to lower bounds on the quantum query complexity of evaluating read-many formulas.  We begin by proving a basic lemma on the query complexity of circuits obtained by composition, and then use this lemma (together with known quantum lower bounds) to show that the algorithm of \thm{alg} is optimal.

\subsection{Composition}
\label{sec:comp}

Recent work on the quantum adversary method has shown that quantum query complexity behaves well with respect to composition of functions \cite{Rei09,Rei11}.  Here we apply this property to characterize the query complexity of composed circuits.  By composing circuits appropriately, we construct a circuit whose depth is one less than the sum of the depths of the constituent circuits, but whose query complexity is still the product of those for the constituents.  This construction can be used to give tradeoffs between circuit size and quantum query complexity: in particular, it shows that good lower bounds for functions with large circuits can be used to construct weaker lower bounds for functions with smaller circuits.

First we note some simple transformations that can be applied to a given circuit.  Without decreasing the query complexity or increasing the depth, and at the cost of at most doubling the number of inputs and gates, we can assume that the given circuit is monotone (i.e., consists only of \textsc{and} and \textsc{or} gates, with no \textsc{not} gates), and its topmost gate is a gate of our choosing (either \textsc{and} or \textsc{or}, as desired). 

\begin{lemma}\label{lem:wlog}
Let $f$ be a circuit with $n_f$ inputs, having depth $k$ and size $G$.  Then there is a monotone circuit $f'$ with $2n_f$ inputs, size at most $2G$, depth $k$, and a topmost gate either \textsc{and} or \textsc{or} (as desired), such that $Q(f) \leq Q(f')$.  Furthermore, if $f$ is a formula, $f'$ is also a formula of the same size.
\end{lemma}

\begin{proof}
First observe that any \textsc{not} gates in the circuit $f$ can be pushed to the inputs using De Morgan's laws. This at most doubles the number of gates.  Then, if $f$ has input variables $x_1,\ldots,x_{n_f}$, let $f'$ have the $2n_f$ inputs $x_1, \bar x_1, \ldots, x_{n_f}, \bar x_{n_f}$, so that it is unnecessary to apply \textsc{not} gates to the inputs. Now $Q(f) \leq Q(f')$ because any algorithm for $f'$ can be converted to an algorithm for $f$ using the same number of queries.

To switch the output gate from $\textsc{and}$ to \textsc{or} or vice versa, simply consider the negation of $f$.  The resulting function has the same query complexity, but the output gate is switched.
\end{proof}

Now we are ready to prove the composition lemma.  

\begin{lemma}\label{lem:comp}
Let $f$ be a circuit with $n_f$ inputs, having depth $k_f$ and size $G_f$; and let $g$ be a circuit with $n_g$ inputs, having depth $k_g$ and size $G_g$.  Then there exists a circuit $h$ with $n_h = 4 n_f n_g$ inputs, having depth $k_h = k_f + k_g - 1$ and size $G_h \leq 2 G_f + 4 n_f G_g$, such that $Q(h) = \Omega(Q(f)Q(g))$.  Furthermore, if $f$ is a formula and $k_g=1$, then $h$ is a formula of size $S_h = S_f S_g$.
\end{lemma}

\begin{proof}
By \lem{wlog}, we can assume that $f$ and $g$ are monotone at the cost of replacing $n_f$ by $2n_f$, $G_f$ by $2G_f$, $n_g$ by $2n_g$, and $G_g$ by $2G_g$, but with no change to the depths of the circuits.  Furthermore, we can assume without loss of generality that the gates of $f$ at level $k_f$ are of the same type as the top gate of $g$.  These assumptions ensure that when we compose $f$ with $g$, the top gate of $g$ can be merged with the gates of $f$ at level $k_f$, since all such gates are of the same type.

Now let $h=f \circ (g,\ldots,g)$ be the composition of $f$ with $n_f$ copies of $g$, i.e.,
\begin{align}
  h(x_1,\ldots,x_{n_f n_g}) = f\(g(x_1,\ldots,x_{n_g}),\ldots,g(x_{n_f n_g-n_g+1},\ldots,x_{n_f n_g})\).
\end{align}
By combining adjacent gates of the same type, we have $k_h = k_f + k_g - 1$.  The expressions for the number of inputs $n_h$ and the size $G_h$ are immediate.

Theorem 1.5 of \cite{Rei11} shows that the quantum query complexity of the composed function is simply the product of the individual query complexities, up to some constant factor.

If $k_g=1$, then $g$ is simply an \textsc{and} or \textsc{or} gate, so it clearly can be composed with a formula $f$ to give a formula $h$.  Since each of the $S_f$ inputs of $f$ (counted with multiplicity) gives rise to $S_g$ inputs of $g$ (again, counted with multiplicity), the size of this formula is simply $S_h = S_f S_g$.
\end{proof}

Observe that in general, the fanout of most gates in $h$ is inherited from the corresponding gates in $f$ and $g$, except that the gates at level $k_f$ of $f$ are combined with the top gate of $g$.  For example, if $f$ is a read-once formula and $g$ is a formula (or if $f$ is a formula and $k_g=1$ as considered above), then $h$ is a formula.

\subsection{Optimality of \texorpdfstring{\thm{alg}}{Theorem \ref*{thm:alg}}}
\label{sec:tight}

In this section we give examples of functions for which the algorithm of \thm{alg} is optimal.  We obtain such functions by composing a formula for the \textsc{parity} function with \textsc{and} gates.

It is well known that the parity of $n$ bits, denoted $\textsc{parity}_n$, can be computed by a formula of size $O(n^2)$. An explicit way to construct this formula is by recursion. The parity of two bits $x$ and $y$ can be expressed by a formula of size $4$: $x \oplus y = (x \wedge \bar{y}) \vee (\bar{x} \wedge y)$. When $n$ is a power of 2, given formulas of size $n^2/4$ for the parity of the first and second half of the input, we get a formula of size $n^2$. When $n$ is not a power of 2 we can use the next largest power of 2 to obtain a formula size upper bound. Thus \textsc{parity} has a formula of size $O(n^2)$. Consequently, the number of gates in this formula is $O(n^2)$. 

This observation combined with \lem{comp} and known quantum lower bounds for \textsc{parity} gives us the following theorem.

\begin{theorem}\label{thm:tight}
For any $n,S,{G}$, there is a read-many formula with $n$ inputs, size at most $S$, and at most ${G}$ gates with bounded-error quantum query complexity $\Omega(\min\{n, \sqrt{S}, n^{1/2} {G}^{1/4}\})$.
\end{theorem}

\begin{proof}
If $\min\{n, \sqrt{S}, n^{1/2} {G}^{1/4}\} = n$ (i.e., $S \ge n^2$ and $G \ge n^2$), then consider the \textsc{parity} function, which has $Q(\textsc{parity}_n) = \Omega(n)$~\cite{BBCMW01,FGGS98}.  Since the formula size and gate count of parity are $O(n^2)$, this function has formula size $O(S)$ and gate count $O(G)$.  By adjusting the function to compute the parity of a constant fraction of the inputs, we can ensure that the formula size is at most $S$ and the gate count is at most $G$ with the same asymptotic query complexity.

In the remaining two cases, we obtain the desired formula by composing a formula for \textsc{parity} with \textsc{and} gates.
We apply \lem{comp} with $f = \textsc{parity}_{m}$ and $g = \textsc{and}_{n/m}$ for some choice of $m$.  The resulting formula has $\Theta(n)$ inputs, size $O(m^2(n/m)) = O(nm)$, and gate count $O(m^2)$.  Its quantum query complexity is $\Omega(m\sqrt{n/m}) = \Omega(\sqrt{nm})$.

If $\min\{n, \sqrt{S}, n^{1/2} {G}^{1/4}\} = \sqrt{S}$ (i.e., $S \le n^2$ and $S \le n \sqrt{G}$), let $m = S/n$.  Then the formula size is $O(S)$ and the gate count is $O(S^2/n^2) \le O(G)$.  By appropriate choice of constants, we can ensure that the formula size is at most $S$ and the gate count is at most $G$.  In this case, the query complexity is $\Omega(\sqrt{S})$.

Finally, if $\min\{n, \sqrt{S}, n^{1/2} {G}^{1/4}\} = n^{1/2} {G}^{1/4}$ (i.e., $G \le n^2$ and $S \ge n \sqrt{G}$), let $m = \sqrt{G}$.  Then the gate count is $O(G)$ and the formula size is $O(n \sqrt{G}) \le O(S)$.  Again, by appropriate choice of constants, we can ensure that the formula size is at most $S$ and the gate count is at most $G$.  In this final case, the query complexity is $\Omega(n^{1/2} G^{1/4})$.
\end{proof}

\section{Constant-depth formulas}
\label{sec:const}

\thm{alg} and \thm{tight} together completely characterize the quantum query complexity of formulas with $n$ inputs, formula size $S$, and gate count ${G}$. However, while there exists such a formula for which the algorithm is optimal, particular formulas can sometimes be evaluated more efficiently.  Thus it would be useful to have a finer characterization that takes further properties into account, such as the depth of the formula. In this section we consider formulas of a given depth, number of inputs, size, and gate count.

Since the algorithm described in \sec{alg} works for formulas of any depth, we know that any depth-$k$ formula with ${G}$ gates can be evaluated with $O(\min\{n, \sqrt{S}, n^{1/2} {G}^{1/4}\})$ queries. However, the lower bound in \sec{lb} uses a formula with non-constant depth. Here we focus on proving lower bounds on constant-depth formulas.

Consider the \textsc{onto} function (defined in \cite{BM10}), which has a depth-3 formula and has nearly maximal query complexity.  For any positive even integer $n$, let $X_n$ be the set of functions from $[2n-2]$ to $[n]$.  Then the function $\textsc{onto}\colon X_n \to \{0,1\}$ has $\textsc{onto}(f) = 1$ iff $f$ is surjective.  To view this as a Boolean function, we can encode the $2n-2$ values $f(i) \in [n]$ in binary, giving a function of $N=(2n-2)\log n$ bits.  We sometimes use a subscript to indicate the number of input bits, writing $\textsc{onto}_N$ where $N=(2n-2)\log n$.

Beame and Machmouchi showed that the query complexity of the \textsc{onto} function is linear in the size of the range, i.e., nearly linear in the number of input bits:

\begin{proposition}[Corollary 6 of \cite{BM10}]\label{prop:onto}
  $Q(\textsc{onto}_N) = \Omega(N/\log N)$.
\end{proposition}

Furthermore, $\textsc{onto}$ has a simple depth-$3$ formula of size $\Theta(n^2 \log n) = \tilde\Theta(N^2)$, namely \cite{BM10}
\begin{align}
  \textsc{onto}(f) = \bigwedge_{j \in [n]} \bigvee_{i \in [2n-2]} \bigwedge_{\ell=0}^{\log_2 n - 1} f(i)_\ell^{j_\ell}
  \label{eq:ontoformula}
\end{align}
where $f(i)_\ell$ is the \th{\ell} bit in the binary encoding of $f(i)$, $j_\ell$ is the \th{\ell} bit in the binary encoding of $j$, and $x^b$ is $x$ if $b=1$ and $\bar x$ if $b=0$.

Now using the \textsc{onto} function instead of \textsc{parity} in the proof of \thm{tight}, we get the following.

\begin{theorem}\label{thm:depth3}
For any $N,S,{G}$, there is a depth-3 formula with $N$ inputs, size at most $S$, and at most ${G}$ gates with quantum query complexity $\tilde\Omega(\min\{N, \sqrt{S}, N^{1/2} {G}^{1/4}\})$.
\end{theorem}

This gives us a matching lower bound, up to log factors, for any $k \geq 3$. Thus we have a nearly tight characterization of the query complexity of evaluating depth-$k$ formulas with $n$ inputs, size $S$, and ${G}$ gates for all $k \ge 3$.
Furthermore, since any depth-1 formula is either the \textsc{and} or \textsc{or} of a subset of the inputs, the query complexity of depth-1 formulas is easy to characterize. Thus we have a characterization for all depths other than depth 2, and it remains to consider the depth-2 case.
Since all depth-2 circuits are formulas, we will refer to depth-2 formulas as depth-2 circuits henceforth and use circuit size (i.e., the number of gates) as the size measure.

There are also independent reasons for considering the query complexity of depth-2 circuits. Improved lower bounds on depth-2 circuits of size $n$ imply improved lower bounds for the Boolean matrix product verification problem. We explain this connection and exhibit new lower bounds for the Boolean matrix product verification problem in \sec{bmpv}. 

Furthermore, some interesting problems can be expressed with depth-2 circuits, and improved upper bounds for general depth-2 circuits would give improved algorithms for such problems.  For example, the graph collision problem~\cite{MSS07} for a graph with $G$ edges can be written as a depth-2 circuit of size $G$. Below, we exhibit a lower bound for depth-2 circuits that does not match the upper bound of \thm{alg}. If there exists an algorithm that achieves the query complexity of the lower bound, this would improve the best known algorithm for graph collision, and consequently the triangle problem \cite{MSS07}.

We begin with lower bounds for depth-2 circuits. The trivial lower bound for depth-2 circuits is $\Omega(\sqrt{n})$ due to the \textsc{or} function.  In this section we obtain better lower bounds for depth-2 circuits using the element distinctness problem. 

In the element distinctness problem, we are given a string $x_1x_2\ldots x_n \in [n]^n$ of length $n$ over an alphabet of size $n$, and we are asked if there are two positions in the string that are equal, i.e., whether there exist $i,j \in [n]$ with $i \neq j$ such that $x_i = x_j$. Solving this problem requires $\Omega(n^{2/3})$~\cite{AS04,Ambainis05,Kutin05} queries to an oracle that returns $x_i$ when queried with $i$.

To express this problem as a Boolean function, we represent the $n$ inputs in binary using ${\log n}$ bits. The size of the input is $N = n \log n$; in terms of $N$, the lower bound is $\Omega((N/\log N)^{2/3})$. Observe that the element distinctness problem can be represented by a depth-2 circuit of size $O(n^3)$. One way to see this is by noting that we can check the condition $x_i = x_j = k$ for any $i,j,k \in [n]$ using 1 \textsc{and} gate and some \textsc{not} gates. Then we just take the \textsc{or} of $\binom{n}{2} n = O(n^3)$ such clauses to check whether two distinct inputs map to the same output. This immediately gives the following:

\begin{theorem}\label{thm:depth2}
There exists a depth-2 circuit of size $O((N/\log N)^3)$ that requires $\Omega((N/\log N)^{2/3})$ quantum queries to evaluate.
\end{theorem}

For the application to Boolean matrix product verification, we need lower bounds on depth-2 circuits with $n$ gates. Such a lower bound is easy to obtain from \thm{depth2} using \lem{comp}:

\begin{corollary}\label{cor:depth2}
There exists a depth-2 circuit on $n$ inputs of size $n$ that requires $\tilde\Omega(n^{5/9}) = \Omega(n^{0.555})$ quantum queries to evaluate.
\end{corollary}
\begin{proof}
Let $f$ be the element distinctness function. We know that $k_f = 2$ and $G_f = \tilde O(n_f^3)$. Let $g$ be the \textsc{and} function, which has $k_g = 1$ and $G_g = 1$. Composing these functions using \lem{comp} gives a depth-2 circuit $h$ with $n_h = 4 n_f n_g$, $G_h = \tilde O(n_f^3)$, and $Q(h) = \Omega(n_f^{2/3} \sqrt{n_g})$.  Choosing $n_f = n_h^{1/3}$ and $n_g = n_h^{2/3}$, we have $G_h = \tilde O(n_h)$ and $Q(h) = \tilde\Omega(n^{5/9})$.  Adjusting $n_f$ and $n_g$ by logarithmic factors, we can set $G_h = n_h$ with only a logarithmic adjustment to the query complexity.
\end{proof}

The results of this section also allow us to characterize (up to log factors) the quantum query complexity of bounded-fanout \ACzero circuits, which we call \ACzerob circuits. \ACzerob circuits are like \ACzero circuits where the gates are only allowed to have $O(1)$ fanout, as opposed to the arbitrary fanout that is allowed in \ACzero. Note that as complexity classes, \ACzero and \ACzerob are equal, but the conversion from an \ACzero circuit to an \ACzerob circuit will in general increase the circuit size.

The reason that formula size bounds apply to \ACzerob circuits is that an \ACzerob circuit can be converted to formula with a constant factor increase in size. (However, this constant depends exponentially on the depth.) Thus we have the following corollary:

\begin{corollary}\label{cor:ac0b}
Any language $L$ in \ACzerob has quantum query complexity $O(\min\{n, n^{1/2} G^{1/4}\})$ where $G(n)$ is the size of the smallest circuit family that computes $L$. Furthermore, for every $G(n)$, there exists a language in \ACzerob that has circuits of size $G(n)$ and that requires $\tilde\Omega(\min\{n, n^{1/2} G^{1/4}\})$ quantum queries to evaluate.
\end{corollary}

\section{Applications}
\label{sec:app}

\subsection{Boolean matrix product verification}
\label{sec:bmpv}

A decision problem closely related to the old and well-studied matrix multiplication problem is the matrix product verification problem. This is the task of verifying whether the product of two matrices equals a third matrix. The Boolean matrix product verification (\textsc{bmpv}) problem is the same problem where the input comprises Boolean matrices and the matrix product is performed over the Boolean semiring, i.e., the ``sum'' of two bits is their logical \textsc{or} and the ``product'' of two bits is their logical \textsc{and}.

More formally, the input to the problem consists of three $n \times n$ matrices $A$, $B$ and $C$. We have to determine whether $C_{ij} = \bigvee_k A_{ik}\wedge B_{kj}$ for all $i, j \in [n]$. Note that the input size is $3n^2$, so $O(n^2)$ is a trivial upper bound on the query complexity of this problem. Buhrman and \v{S}palek~\cite{BS06} show that \textsc{bmpv} can be solved in $O(n^{3/2})$ queries. This bound can be obtained by noting that checking the correctness of a single entry of $C$ requires $O(\sqrt{n})$ queries, so one can use Grover's algorithm to search over the $n^2$ entries for an incorrect entry.  Another way to obtain the same bound is to show that there exists a formula of size $O(n^3)$ that expresses the statement $AB = C$ and then apply \thm{formeval}.

While Buhrman and \v{S}palek do not explicitly state a lower bound for this problem, their techniques yield a lower bound of $\Omega(n)$ queries. This leaves a gap between the best known upper and lower bounds. The following theorem summarizes known facts about Boolean matrix product verification.

\begin{theorem}\label{thm:bmpv}
If $A$, $B$ and $C$ are $n \times n$ Boolean matrices available via oracles for their entries, checking whether the Boolean matrix product of $A$ and $B$ equals $C$, i.e., checking whether $C_{ij} = \bigvee_k A_{ik}\wedge B_{kj}$ for all $i,j \in [n]$, requires at least $\Omega(n)$ quantum queries and at most $O(n^{3/2})$ quantum queries.
\end{theorem} 

We use the results of the previous section to improve the lower bound to $\Omega(n^{1.055})$. The first step is to show how the problem of evaluating depth-2 circuits relates to the Boolean matrix product verification problem.

Let the Boolean vector product verification problem for a given $n \times n$ matrix $A$ ($\textsc{bvpv}_A$) be the problem of deciding if a given vector $v$ satisfies $Av = 1$, where 1 is the all-ones vector of length $n$. Note that $A$ is part of the specification of the problem and not part of the input. Let the Boolean function computed by this problem be denoted as ${\textsc{bvpv}}_A(v)$, where $v$ is a Boolean vector of size $n$.

Since $v$ is a vector of size $n$, the query complexity of this problem, $Q({\textsc{bvpv}}_A)$, is upper bounded by $n$. 
Observe that we can write ${\textsc{bvpv}}_A(v)$ as  $\bigwedge_i \bigvee_j A_{ij}v_{j}$, which is a monotone depth-2 circuit with $n$ \textsc{or} gates, 1 \textsc{and} gate, and $n$ input variables $v_i$. More interestingly, every monotone depth-2 circuit with $n$ \textsc{or} gates, 1 \textsc{and} gate, and $n$ input variables corresponds to $\textsc{bvpv}_A$ for some matrix $A$. 

It follows that lower bounds for linear-sized depth-2 circuits also yield lower bounds for this problem. From \cor{depth2}, we know that there exists a depth-2 circuit with $n$ gates that requires $\Omega(n^{0.555})$ queries to evaluate. Without loss of generality, we can assume (by \lem{wlog}) that the circuit is monotone and that its top gate is an \textsc{and} gate.
Thus there exists a matrix $A$ such that $Q({\textsc{bvpv}}_A) = \Omega(n^{0.555})$.

The next lemma shows us how the \textsc{bvpv} problem is related to the \textsc{bmpv} problem. 

\begin{lemma}\label{lem:bvpv}
For any $n\times n$ matrix $A$, 
$Q({\textsc{bmpv}}) = \Omega(\sqrt{n}\thinspace Q({\textsc{bvpv}}_A))$
\end{lemma}

\begin{proof}
We prove a lower bound for the special case of \textsc{bmpv} where $C$ is the all ones matrix, $J$. In this case, checking whether $AB = J$ is equivalent to checking whether $Ab_i = 1$ for all $i$ where $b_i$  denotes the \th{i} column of $B$. Indeed, we can think of this problem as $n$ independent instances of the ${\textsc{bvpv}}_A$ problem: the output of this special case of \textsc{bmpv} with the first matrix being $A$ is 1 if and only if all $n$ instances of ${\textsc{bvpv}}_A$ output 1. In other words, the \textsc{bmpv} problem for a fixed $A$ and $C=J$ is just  $\textsc{bvpv}_A(b_1) \wedge \ldots \wedge \textsc{bvpv}_A(b_n) = \textsc{and}_n \circ (\textsc{bvpv}_A, \ldots, \textsc{bvpv}_A) (b_1, \ldots, b_n)$. 

As in \sec{comp}, we use Theorem 1.5 of \cite{Rei11} to conclude the quantum query complexity of the \textsc{bmpv} problem for a fixed $A$ and $C=J$ is $\Omega(Q(\textsc{and}_n) Q({\textsc{bvpv}}_A))$. Since the general \textsc{bmpv} problem can only be harder than this special case, we get the desired lower bound.
\end{proof}

Using this lemma and the lower bound obtained earlier,  $Q({\textsc{bvpv}}_A) = \tilde \Omega(n^{5/9}) = \Omega(n^{0.555})$, we get the main result of this section:

\begin{theorem}
The bounded-error quantum query complexity of the Boolean matrix product verification problem is $\tilde\Omega(n^{19/18}) = \Omega(n^{1.055})$.
\end{theorem}

\subsection{Applications to classical circuit complexity}
\label{sec:classical}

In this section we present some classical applications of our results. In particular, we prove lower bounds on the number of gates needed in any formula representing certain functions. The main tool we use is the following corollary of \thm{alg}.

\begin{corollary}\label{cor:alg}
For a function $f$ with $n$ inputs and quantum query complexity $Q(f)$, any (unbounded-fanin) formula representing $f$ requires $\Omega(Q(f)^4/n^2)$ gates.
\end{corollary}

Almost immediately, this implies that functions such as \textsc{parity} and \textsc{majority}, which have quantum query complexity of $\Omega(n)$, require $ \Omega(n^2)$ gates to be represented as a formula. Similarly, this implies that functions with query complexity $\Omega(n^{3/4})$, such as \textsc{graph connectivity}~\cite{DHHM06}, \textsc{graph planarity}~\cite{AIN+08}, and \textsc{hamiltonian cycle}~\cite{BDF+03}, require formulas with $ \Omega(n)$ gates.

To compare this with previous results, it is known that \textsc{parity} requires formulas of size $\Omega(n^2)$~\cite{Khr71}. This result is implied by our result since the number of gates is less than the size of a formula.

We can also use these techniques to address the question ``Given a constant-depth circuit of size $G$, how efficiently can this circuit be expressed as a formula?''
The best result of this type that we are aware of is the following: There exists a constant-depth circuit of linear size such that any formula expressing the same function has size at least $n^{2-o(1)}$. The result appears at the end of Section 6.2 in \cite{Juk12}, where the function called $V^\textsc{or}(x,y)$ is shown to have these properties. Indeed, the function has a depth-3 formula with $O(n)$ gates. The idea of using such functions, also called universal functions, is attributed to Nechiporuk~\cite{Nec66}.

We construct an explicit constant-depth circuit of linear size that requires $\Omega(n^{2-\epsilon})$ gates to be expressed as a formula. To the best of our knowledge, this result is new.  Our result is incomparable to the previous result since we lower bound the number of gates, which also lower bounds the formula size, but we use a constant-depth circuit as opposed to a depth-3 formula.

Improving our lower bound to $\Omega(n^2)$ seems difficult, since we do not even know an explicit constant-depth circuit of linear size that requires formulas of \emph{size} (as opposed to number of gates) $\Omega(n^2)$, which is a weaker statement. In fact, we do not know any explicit function in \ACzero with a formula size lower bound of $\Omega(n^2)$ (for more information, see~\cite{CST1}).

\begin{theorem}\label{thm:linsize}
For every $\epsilon>0$, there exists a constant-depth unbounded-fanin circuit (i.e., an \ACzero circuit) of size $O(n)$ such that any (unbounded-fanin) formula computing the same function must have $\Omega(n^{2-\epsilon})$ gates.
\end{theorem}
\begin{proof}
The aim is to construct a function in \ACzero with linear size and nearly maximal quantum query complexity, and then to apply \cor{alg}. We know that there exists a depth-3 function, $\textsc{onto}$, with query complexity $\Omega(n/\log n)$ and size $O(n^2/\log n)$ (\prop{onto}). Composing the $\textsc{onto}$ function with itself using \lem{comp} gives a depth-5 circuit of smaller size with a similar query complexity lower bound.

In general, if we have a function $f$ with a circuit of depth $k$, size $O(n^r)$ for some $1 \leq r \leq 2$, and query complexity $\Omega(n/\log^c n)$, then we can construct a function $f'$ with a circuit of depth $2k-1$, size $O(n^{{r^2}/(2r-1)})$, and query complexity $\Omega(n/\log^{2c} n)$. This is achieved by composing the function $f$ on $m$ inputs with $f$ on $n/m$ inputs to get a new function $f'$ on $n$ inputs. The size of the resulting circuit is $O(m^r + n^r/m^{r-1})$. To make the two terms equal, we choose $m^{2r-1} = n^r$, which gives the claimed size. The query complexity of $f'$ is $\Omega((m / \log^c m)((n/m) / \log^c (n/m))) = \Omega(n / \log^{2c} n)$.

Now we can iterate this construction to get smaller circuits with about the same query complexity. Since the size of the circuit decreases from $O(n^r)$ to $O(n^{{r^2}/(2r-1)})$, the circuit size approaches $O(n)$ as the number of iterations increases. Clearly, for any $\delta>0$, we can reach a circuit of size $O(n^{1+\delta})$ with a constant number of iterations. The resulting function has constant depth and query complexity $\Omega(n / \log^{c} n)$, for some constant $c$ that depends on $\delta$. Now we can define a new function $g$ on $n$ inputs that is this function acting on the first $n^{1/(1+\delta)}$ bits. Clearly this function has a linear-size circuit, and its query complexity is $\Omega(n^{1/(1+\delta)} / \log^{c} n)$, which is $\Omega(n^{1-\epsilon'})$ for some $\epsilon'>0$. Since $\delta>0$ could be chosen arbitrarily, we can achieve any $\epsilon'>0$ in this step. \cor{alg} now completes the proof.
\end{proof}

\section{Conclusions and open problems}
\label{sec:conc}

We have given a tight characterization of the query complexity of read-many formulas in terms of their number of inputs $n$, formula size $S$, and gate count $G$. In particular, we showed that the query complexity of evaluating this class of formulas is $\Theta(\min\{n, \sqrt{S}, n^{1/2}G^{1/4}\})$. Our results suggest several new avenues of research, looking both toward refined characterizations of the query complexity of evaluating formulas and toward a better understanding of the quantum query complexity of evaluating circuits.

In \sec{const} we showed that our query complexity bounds are nearly tight for all formulas of a given depth except for the case of depth 2. We made partial progress on this remaining case by giving a depth-2 circuit of size $n$ that requires $\Omega(n^{0.555})$ queries, whereas our algorithm gives an upper bound of $O(n^{0.75})$.

We expect that it should be possible to improve the lower bound for linear-size depth-2 circuits.  We have considered candidate depth-2 circuits, based on affine or projective planes, that seem difficult to evaluate.  Here we describe a circuit based on a projective plane, which is a set $P$ of points and a set $L$ of lines such that any two distinct points are on a unique line, any two distinct lines intersect at a unique point, and there exist four points with no three of them on the same line. A projective plane of order $q$ consists of $n=q^2+q+1$ points and $n$ lines, where each line contains $q+1$ points and each point is on $q+1$ lines. We write $i\in \ell$ to indicate that the point $i \in P$ is on the line $\ell \in L$. Consider the depth-2 circuit
\begin{equation}
\label{eq:proj}
\bigvee_{\ell\in L}\bigwedge_{i\in \ell} x_i
\end{equation}
where $x_i$ is a bit assigned to the point $i \in P$. This circuit has $n$ variables and $n+1$ gates. Clearly, the 1-certificate complexity of this formula is $q+1$; it can also be shown that when $q$ is a square, the 0-certificate complexity is $q^{3/2}+1$ (see for example \cite{BE08}). Thus the certificate complexity barrier \cite{SS06,Zha05} only rules out proving a better lower bound than $\Omega(n^{5/8})=\Omega(n^{0.625})$ using the adversary method with positive weights.  However, we are not aware of any lower bound better than $\Omega(n^{3/8})$, or any upper bound better than the $O(n^{3/4})$ bound of \thm{alg}.

While we have made progress in understanding the quantum query complexity of evaluating read-many formulas, we would also like to understand the query complexity of evaluating general circuits.  It would be interesting to find upper and lower bounds on the query complexity of evaluating circuits as a function of various parameters such as their number of inputs, gate count, fanout, and depth.  In particular, the graph collision problem can also be expressed using a circuit of depth $3$ and linear size (in addition to the naive depth-2 circuit of quadratic size mentioned in \sec{const}), so it would be interesting to focus on the special case of evaluating such circuits.

\section*{Acknowledgments}

We thank Noam Nisan for the proof of \prop{counting} (via the website \href{http://cstheory.stackexchange.com}{cstheory.stackexchange.com}).  R.~K.\ thanks Stasys Jukna for helpful discussions about circuit complexity.

This work was supported in part by MITACS, NSERC, the Ontario Ministry of Research and Innovation, QuantumWorks, and the US ARO/DTO.
This work was done while S.~K.\ was visiting the Institute for Quantum Computing at the University of Waterloo.
S.~K.\ also received support from NSF Grant No. DGE-0801525,
\emph{IGERT: Interdisciplinary Quantum Information Science and Engineering}, and from the U.S.\ Department of Energy under cooperative research agreement Contract Number DE-FG02-05ER41360.

\end{document}